\documentclass{article}

\usepackage[preprint,nonatbib]{neurips_2022}

\usepackage[utf8]{inputenc} 
\usepackage[T1]{fontenc}    
\usepackage{hyperref}       
\usepackage{url}            
\usepackage{booktabs}       
\usepackage{amsfonts}       
\usepackage{nicefrac}       
\usepackage{microtype}      
\usepackage{xcolor}         

\usepackage{amsmath}
\usepackage{amsthm}
\usepackage{amsfonts}
\usepackage{amssymb}
\usepackage{mathtools} 
\usepackage{bbm} 
\usepackage{bm} 
\usepackage{graphicx}
\usepackage[ruled,vlined,linesnumbered]{algorithm2e}

\DeclareMathOperator*{\argmin}{arg\,min}
\DeclareMathOperator*{\argmax}{arg\,max}

\newtheorem{definition}{Definition}
\newtheorem{proposition}{Proposition}
\newtheorem{lemma}{Lemma}
\newtheorem{theorem}{Theorem}

\newtheorem{assumption}{Assumption}

\newcommand{\1}{\mathbbm{1}}

\def\ddefloop#1{\ifx\ddefloop#1\else\ddef{#1}\expandafter\ddefloop\fi}
\def\ddef#1{\expandafter\def\csname bb#1\endcsname{\ensuremath{\mathbb{#1}}}}
\ddefloop ABCDEFGHIJKLMNOPQRSTUVWXYZ\ddefloop

\def\ddef#1{\expandafter\def\csname c#1\endcsname{\ensuremath{\mathcal{#1}}}}
\ddefloop ABCDEFGHIJKLMNOPQRSTUVWXYZ\ddefloop

\def\ddef#1{\expandafter\def\csname v#1\endcsname{\ensuremath{\boldsymbol{#1}}}}
\ddefloop ABCDEFGHIJKLMNOPQRSTUVWXYZabcdefghijklmnopqrstuvwxyz\ddefloop

\def\ddef#1{\expandafter\def\csname v#1\endcsname{\ensuremath{\boldsymbol{\csname #1\endcsname}}}}
\ddefloop {alpha}{beta}{gamma}{delta}{epsilon}{varepsilon}{zeta}{eta}{theta}{vartheta}{iota}{kappa}{lambda}{mu}{nu}{xi}{pi}{varpi}{rho}{varrho}{sigma}{varsigma}{tau}{upsilon}{phi}{varphi}{chi}{psi}{omega}{Gamma}{Delta}{Theta}{Lambda}{Xi}{Pi}{Sigma}{varSigma}{Upsilon}{Phi}{Psi}{Omega}{ell}\ddefloop

\usepackage[backend=bibtex,style=numeric]{biblatex}
\newcommand{\syremoved}[1]{}

\newcommand{\Rad}{\text{Rad}}
\newcommand{\optp}{p^*}
\newcommand{\oraclep}{p}
\newcommand{\dualg}{\bar g_{\gamma, c}}
\newcommand{\fat}{\text{fat}}
\newcommand{\Fbound}{(R\bar x + R)}

\title{Optimal Efficiency-Envy Trade-Off via Optimal Transport}

\author{%
  Steven Yin\\ 
  Department of Industrial Engineering and Operations Research\\
  Columbia University\\
  New York, NY 10027\\
  \texttt{sy2737@columbia.edu} \\
  \And
  Christian Kroer\\ 
  Department of Industrial Engineering and Operations Research\\
  Columbia University\\
  New York, NY 10027\\
  \texttt{christian.kroer@columbia.edu} \\
}

\begin{document}

\maketitle

\begin{abstract}
	We consider the problem of allocating a distribution of items to $n$
	recipients where each recipient has to be allocated a fixed, prespecified
	fraction of all items, while ensuring that each recipient does not experience
	too much envy.  We show that this problem can be formulated as a variant of
	the semi-discrete optimal transport (OT) problem, whose solution structure
	in this case has a concise representation and a simple geometric
	interpretation.  Unlike existing literature that treats envy-freeness as a
	hard constraint, our formulation allows us to \emph{optimally} trade off
	efficiency and envy continuously.  Additionally, we study the statistical
	properties of the space of our OT based allocation policies by showing a
	polynomial bound on the number of samples needed to approximate the optimal
	solution from samples.  Our approach is suitable for large-scale fair
	allocation problems such as the blood donation matching problem, and we show
	numerically that it performs well on a prior realistic data simulator.
\end{abstract}
\section{Introduction}
\label{sec: introduction}

In this work, we focus on the problem of finding an allocation policy that
divides a pool of items, represented by a distribution $\cD$, to $n$ recipients,
under the constraint that each recipient $i$ must be allocated a pre-specified
fraction $\optp_i$ of the items, where $\optp_i\in(0,1), 1^\top \optp = 1$ is an
input to the problem that characterizes the priority of each recipient. We refer
to $\{\optp_i\}_{i=1}^n$ as the target matching distribution. In addition to
this matching distribution constraint, we also require that the recipient's envy,
which we will define formally later, be bounded. Unlike the existing resource
allocation literature, where envy-free is either treated as a hard constraint or
not considered at all (in divisible settings), we allow the central planner to specify the level of envy
that is tolerated, and find the most efficient allocation given the amount of
envy budget. This allows us to reduce the envy significantly without paying the
full price of fairness with respect to efficiency.

To concretely motivate our model, let us consider the blood donor matching
problem that was first studied in \cite{mcelfresh2020matching}.  The Meta
platform has a tool called \emph{Facebook Blood Donations}, where users who opt
in to receive notifications are notified about blood donation opportunities near
them.  Depending on the user's and the blood bank's specific characteristics
(e.g., age, occupation for the user, and locations, hours for the blood bank),
notifications about different donation opportunities have different
probabilities of resulting in an actual blood donation. The platform would like
to send each user the most relevant notifications (to maximize the total number
of potential blood donations), while maintaining certain fairness criteria for
all the blood banks that participate in this program. Although the platform can
theoretically send each user multiple notifications about multiple blood banks,
for user experience and other practical reasons, this is not done, at least in
the model introduced in \cite{mcelfresh2020matching}.  Therefore in this
problem, users' attention is the scarce resource that the platforms needs to
allocate to different blood banks. The most natural type of fairness criteria in
this setting is perhaps the number of users that received notifications about
each of the blood banks.  For example, it is not desirable to match zero users
to a particular, potentially inconveniently-located, blood bank, even if
matching zero users to this blood bank results in more blood donations in
expectation.  

Another fairness desideratum commonly studied in the
literature is called \emph{envy-freeness}. We say that recipient $A$ envies recipient $B$
if $A$ values $B$'s allocation more than her own. Intuitively, an allocation
that is envy-free--where no agent envies another agent--is perceived to be fair.

This paper considers both of the two fairness criteria mentioned above:
we study a setting where the goal is to maximize social welfare under a matching distribution constraint, while ensuring that each recipient has bounded envy.
We make the following contributions in regards to this problem:
\begin{enumerate}
    \item We formulate it as a constrained version of a semi-discrete optimal
    transport problem and show that the optimal allocation policy has a concise
    representation and a simple geometric structure. 
    This is particularly
    attractive for large-scale allocation problems, due to the fast computation
    of a match given an item. This insight also shines new light on the question of when envy arises, and when the welfare price on envy-freeness is large. 
    \item We propose an efficient stochastic optimization algorithm for this
    problem and show that it has a provable convergence rate of $O(1/\sqrt{T})$.
    \item We investigate the statistical properties of the space of our
    optimal transport based allocation policies by showing a Probably Approximately Correct (PAC)-like sample
    complexity bound for approximating the optimal solution given finite
    samples. 
\end{enumerate}

In Section~\ref{sec: problem formulation} we formally define the problem we are interested in. In Section~\ref{sec: solution structure}, we show that this problem can be
formulated as a semi-discrete optimal transport problem, whose solution has a
simple structure with a nice geometric interpretation. 
Section~\ref{sec: stochastic optimization} develops a practical stochastic optimization algorithm.
In
Section~\ref{sec: learning from samples}, we show that an $\epsilon$-approximate
solution can be found with high probability given $\tilde
O(\frac{n}{\epsilon^2})$ samples, where $n$ is the number of
recipients. Finally, in Section~\ref{sec: experiments} we demonstrate the effectiveness of our approach using both artificial and a semi-real data.

\section{Literature Review}
\paragraph{Blood donation matching.}
\textcite{mcelfresh2020matching} introduced this problem and modeled it as  an
online matching problem, where the matching quality between an user and a blood
bank is assumed to be known to the platform. The model formulation there is complex and their matching policy is rather
cumbersome, requiring a separate parameter for each (donor, recipient) pair.
Compared to their paper, we are able to provide better
structural insights to the problem by utilizing a simpler model that still
captures the most salient part of the problem.

\paragraph{Online Resource Allocation.}

Another strand of work that our paper is closely related to is that of online
resource allocation, especially those with i.i.d. or random permutation input
models. \textcite{agrawal2014dynamic} studied the setting with linear objective and gave
competitive ratio bounds. Then, \textcite{agrawal2014fast} generalized the results to
concave objectives and convex constraints. Later, \textcite{devanur2019near}
improved the approximation ratio bounds and relaxed the input assumptions on the
budgets. \textcite{balseiro2020best} show that online mirror descent on the dual
multipliers does well under both i.i.d. adversarial, and certain non-stationary
input settings.
However, none of theses papers study the envy-free criterion.
Recently, \textcite{balseiro2021regularized} studied an online resource allocation
problem with fairness regularization. Although the authors did not explicitly
study envy regularization, their regularization framework can be modified to
accommodate envy regularization. However, like all the other papers mentioned in
this paragraph, the offline solution is used as the benchmark to measure regret,
but no explicit solution is given to the offline problem. Our analysis focuses
on the offline problem, and draws an explicit connection to optimal transport,
which allowed us to provide a novel PAC-like analysis on the sample complexity
of the problem. In another recent paper, \textcite{sinclair2021sequential} studied
the trade-off between minimizing envy and minimizing waste, which refers to un-allocated resources. Despite close similarity between our titles, their offline benchmark
is the standard Eisenberg-Gale program, which is envy-free, but does not address
the welfare cost of achieving envy-freeness. 
\paragraph{Fair Division.}
Envy is a popular concept studied in the fair division literature. A large body of these papers are formulated as a
cake-cutting problem (\cite{robertson1998cake,chen2013truth,mossel2010truthful})
where the resources are modeled as an interval and the agents' valuations are
represented as functions on this interval. 
\textcite{caragiannis2009efficiency} provide a analysis on the worst case efficiency loss due to the envy-freeness constraint. Later \textcite{cohler2011optimal} design algorithms for computing optimal envy-free cake cutting allocations under different relatively simple classes of valuation functions. 
Unlike these papers that focus on the hard constraint of zero envy, we treat the allowable envy as a parameter, and find the most
efficient solutions subject to the desired amount of envy.
In indivisible settings, there are some related concepts of envy that does not require a strict zero envy. For instance, \textcite{budish2011combinatorial} proposes an approximate competitive equilibrium from equal incomes approach that  achieves envy free up to 1 item (EF1). \textcite{caragiannis2019unreasonable} introduce the concept of envy free up to the Least Valued Good, which is a stronger version of EF1. Although these papers do allow a small amount of envy, these concepts are introduced mainly to circumvent the impossibilities introduced by the indivisible settings, not to allow central planners to control the level of envy tolerance. Finally, \textcite{donahue2020fairness} also studies a setting where the central planner can set his own tolerance of (un)fairness. However their setting is quite different from ours, as they consider the allocation of fixed amount of identical resource, where as we assume that recipients have different valuations for the different items.

\paragraph{Optimal Transport.} OT has been applied before in resource allocation settings in the economics literature (see \cite{galichon2018optimal} for a survey). 
For example, the Hotelling location model with a continuous mass of consumers can be solved with the same assignment procedure as the one we consider for our allocation problem \emph{without} envy constraints. The more general method of assigning points in space to a finite set of sites via this procedure was developed by \textcite{aurenhammer1998minkowski}.
\textcite{scetbon2021equitable} consider an OT setting with equitability (every agent's final utility is the same), which is a different fairness criteria from envy, and also not adjustable like in our setting. To the best of our knowledge, no existing application of OT models the envy constraints that we consider here.

\subsection{Background on Optimal Transport}
Since our work draws an explicit connection to optimal transport (OT), we provide a
summary of key OT results here.  Let $\alpha, \beta$ be two probability measures on
the metric spaces $\cX, \cY$ respectively. We define $\Pi(\alpha, \beta)$ as the
set of joint probability measures on $\cX\times \cY$ with marginals $\alpha$ and
$\beta$. The \emph{Kantorovich formulation} of the optimal transport problem
\cite{kantorovich1942on} can be written as
\begin{align}
     L(\alpha, \beta)\coloneqq \min_{\pi \in \Pi(\alpha, \beta)}\int_{\cX \times \cY} c(x, y) d\pi(x, y)
    \label{eq: optimal tranport primal}
\end{align}
where $c(x, y)$ is the cost associated with ``moving'' $x$ to $y$.  This is
called a transportation problem because the conditional probability $\pi(y|x)$
specifies a transportation plan for moving probability mass from $\cX$ to $\cY$. Note that $\pi(x, y) = \pi(y|x)d\alpha(x)$.
If $\beta$ is a discrete measure, i.e. $\cY$ is finite, then it is known
\cite{aude2016stochastic} that the dual to \eqref{eq: optimal tranport primal}
can be written as (here we abuse the notation $\beta$ to also represent the
vector of probability masses, where $\beta_i$ is the probability mass on point
$y_i$):
\begin{equation}
    \max_{g\in\bbR^n} \cE(g) \coloneqq \sum_{i\in[n]}\left[\int_{\bbL_{y_i}(g)} c(x, y_i)-g_i\,d\alpha(x)\right] + g^\top \beta
    \label{eq: optimal transport dual}
\end{equation}
where $n=|\cY|$, and $\bbL_{y_i}$ is what is sometimes referred to as the \emph{Laguerre
cell}: 
\begin{equation}
\bbL_{y_i}(g) = \left\{x\in\cX: \forall i\neq j, c(x, y_i) - g_i\leq c(x, y_j)-g_{j} \right\}
\label{eq: laguerre cell}
\end{equation}
\begin{proposition}[Proposition~2.1 \cite{aude2016stochastic}]
    If $\alpha$ is a continuous measure, and $\beta$ a discrete measure, then
    $L(\alpha, \beta) = \max\limits_{g}\cE(g)$, and
    the optimal solution $\pi$ of \eqref{eq: optimal tranport primal} is given by the partition $\left\{\bbL_{y_i}(g^*), i\in
    [n]\right\}$, i.e. $d\pi(x, y_i)=d\alpha(x)$ if $x\in \bbL_{y_i}(g^*)$, $0$ otherwise.
    \label{thm: duality of semi-discrete OT}
\end{proposition}

\section{Problem Formulation}
\label{sec: problem formulation}
There is a set of $n$ recipients $\cY$. There is a ``pool'' of items,
represented by a distribution $\alpha$ over $\cX\subseteq [0, \bar x]^n$. Each
random draw from this distribution $X\sim \alpha$ is a vector representing the
$n$ recipients' valuations of this item. The goal is to maximize the expected
matched utilities of the recipients, while maintaining the constraint that the
recipient $y_i$ is matched $\optp_i$ fraction of the times in expectation.  Here
$\{\optp_i\}_{i=1}^n$ is called the target matching distribution, which
intuitively represents recipients' importance. Note that the constrains here are satisfied \emph{in expectation}, which are sometimes referred to as ``ex-ante'' guarantees.  The reason why we consider ex-ante guarantees has to do with the type of application we're interested in. Our motivating example is concerned with recommending hundreds of millions of users to different blood banks. In such settings, even if we want the constraints to hold ex-post, the large-scale nature of the problem and the law of large numbers means that in-expectation guarantees translate into something that is very close to holding ex-post. That is why for internet platform problems, requiring constraints to hold in expectation is a  standard setup (see for example the literature on budget constraints in ad auctions, where this is the case \cite{balseiro2015repeated}). A matching policy $\pi$ takes a
valuation vector and maps it (potentially with randomness) to one of the $n$
recipients. Let $\pi(y|x)$ denote the probability of matching the item to $y$
given valuation vector $x$. The basic problem formulation is to solve the
following optimization problem:
\begin{align}
    \max\limits_{\pi}\,&\bbE_{X\sim\alpha}\left[ \sum_{i=1}^n X_i\pi(y_i|X)]\right]\label{prob: basic optimal matching policy problem}\\ 
    s.t.\quad & \bbP\left[\pi(y_i|X)\right] = \optp_{i}\quad \forall i\in[n]\nonumber
\end{align}
WLOG, we assume that $\optp_i>0$ for all $i$. An example of such problem can be see in Figure~\ref*{fig: 2D
example}, where $\alpha$ is a distribution over the unit square, and the goal is
to partition the square into blue and orange regions (given to $A$ and $B$
respectively) such that each region covers the desired $p^*_A$, $p^*_B$
probability mass. Note that the orange and blue regions are allowed to over lap
(probabilistic partition), and that the boundary does not have to be linear as
illustrated in the figure. As we will show later in Section~\ref{sec: solution
structure}, despite the large design space permitted by the formulation in
\eqref{prob: basic optimal matching policy problem}, we can in fact focus on a
much smaller design space. 

In resource allocation problems, it is often the case that we care not just
about efficiency (the sum of all recipients'
utilities), but also other fairness criteria.  One of the most commonly studied
fairness criteria is envy-freeness. Agent $y_i$ envies another agent $y_j$ if agent
$y_i$ values the allocation given to $y_j$ more (after adjusting for their priority
weights). We can formally define agent $y_i$'s envy as 
\begin{equation}
    \label{eq: envy definition}
    Envy(y_i) = \max_j \bbE_{\alpha}\left[ X_i\pi(y_j|X) \frac{\optp_i}{\optp_j} - X_i\pi(y_i | X)\right]
\end{equation}
Instead of the vanilla formulation in \eqref{prob: basic optimal matching policy
problem}, we consider the following more general formulation:
\begin{align}
    \max\limits_{\pi}\,&\bbE_{X\sim\alpha}\left[ \sum_{i=1}^n X_i\pi(y_i|X)]\right]\label{prob: optimal matching policy problem}\\ 
    s.t.\quad & \bbP\left[\pi(y_i|X)\right] = \optp_i\quad \forall i\in[n]\nonumber\\
    & Envy(y_i)\leq \lambda_i\quad\forall i\nonumber
\end{align}
Most existing literature focuses on finding
allocations such that $Envy(y_i)$ is at most $0$ for every $y_i$. This can be
a very restrictive constraint, often satisfied at the cost of reducing
efficiency by a significant amount (This reduction is sometimes referred to
as the Cost-of-Fairness). We take a different approach, and allow the central
planner to set non-negative constraints on envy. Note that since we are motivated by internet-scale problems such as allocating hundreds of millions of users to donation centers, we focus on the ex-ante guarantees on the constraints.

\section{Optimal Solution Structure}
\label{sec: solution structure}
The space of feasible solutions for \eqref{prob: optimal matching policy problem} is
large, which makes the problem difficult to optimize directly. However we can
use the tools from OT to reduce the search space to something
with much more structure.  The key observation is that \eqref{prob: optimal
matching policy problem} can be formulated as variation of the semi-discrete
optimal transport problem given in Equation~\eqref{eq: optimal tranport primal}. 

Let's first consider the simpler case in \eqref{prob: basic optimal matching
policy problem} where there are no envy constraints.  In this case, the problem
can be stated in the form of \eqref{eq: optimal tranport primal} as follows: the
cost function is the negative utility of the matched recipient $c(x, y_i) =
-x_i$, the $\beta$ measure is the discrete measure
$\sum_{i=1}^n\optp_i\delta_{y_i}$, and the matching policy $\pi(y|x)$ in
\eqref{prob: optimal matching policy problem} is exactly the conditional
probability of the joint distribution in \eqref{eq: optimal tranport primal}.
From Theorem~\ref{thm: duality of semi-discrete OT} it follows that the optimal
matching policy is represented by Laguerre cells given in \eqref{eq: laguerre
cell}: $x$ is matched to $y_i$ if $i = \argmin_{k} -x_{k} - g_{k}$.
Note that the dual
variables $g\in \bbR^n$ serve as an ``adjustment'' over the agents' reported
utilities, and the resulting matching policy is simply a greedy policy over this
adjusted valuation vector. 
\begin{figure*}[]
\begin{minipage}{0.5\textwidth}
    \centering
    \includegraphics[width=0.65\textwidth]{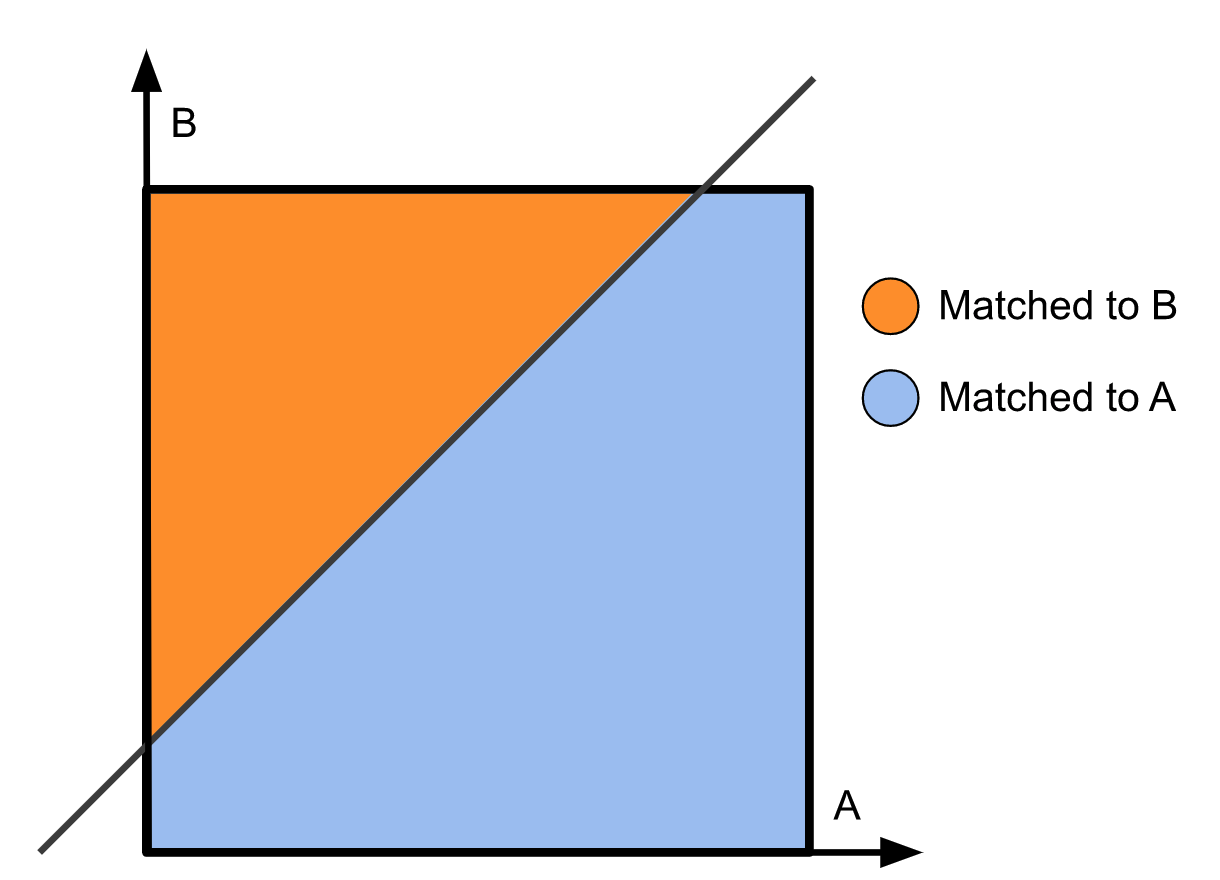}
\end{minipage}%
\begin{minipage}{0.5\textwidth}
    \centering
    \includegraphics[width=0.65\textwidth]{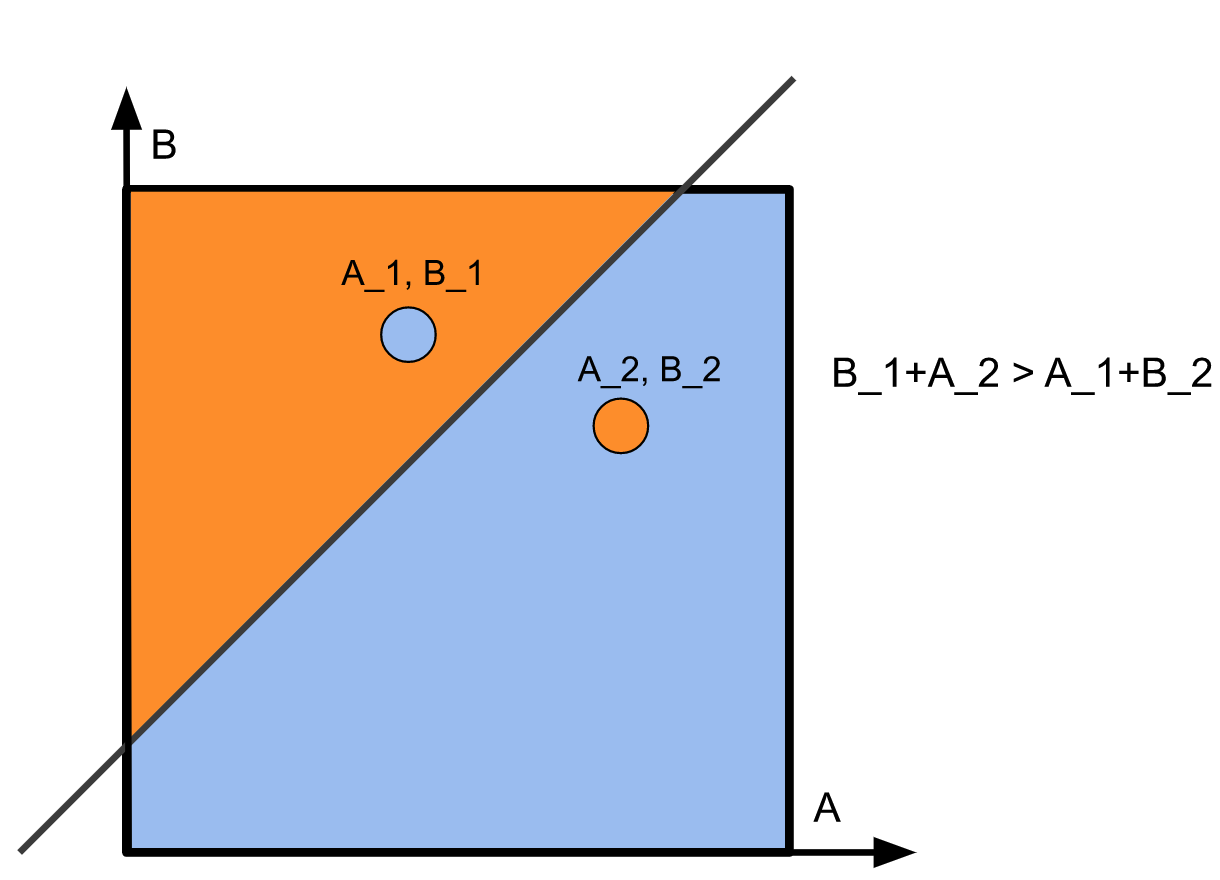}
\end{minipage}
\caption{Left: An illustration of what Laguerre cells look like when $n=2$. Consider any distribution on the support $[0,1]^2$.
The optimal division of the space is to move the diagonal line up or down until
the probability mass contained in the orange region is equal to $\optp$. Right:
A pictorial proof of the optimality of such partition.  Suppose one can find an
$\epsilon$ mass above this line that is matched to $A$, and an $\epsilon$ mass
below the line that is matched to $B$, then switching the assignments of these
two regions increases the matched weights because $B_1+A_2 > A_1+B_2$.}
\label{fig: 2D example}
\end{figure*}

Geometrically, each Laguerre cell is simply the intersection of
half-spaces: $\bbL_i(g)=\cap_{k} \{ x: x_i+g_i \geq x_k + g_k\}$. To visualize this
better, consider the simple setting with two recipients $A, B$ where their
valuations for an item is a joint distribution supported on $[0,1]^2$.
Suppose we want to match $\optp$ fraction of the items to recipient $B$.
Figure~\ref{fig: 2D example} gives a proof-by-picture that the optimal strategy
is to divide the space up with a slope-$1$ diagonal line such
that the probability mass lying above the line is equal to $\optp$. This
geometric interpretation of the matching policy plays a crucial role in getting 
us a sample complexity bound later in Section~\ref{sec: learning from samples}.

With this geometric interpretation of the solution space in mind, let us
consider the more general case with envy constraints as formulated in
\eqref{prob: optimal matching policy problem}.  The envy constraints can be
added to the OT problem in \eqref{eq: optimal tranport primal} like so:
\begin{align}
    L(\alpha, \beta, \lambda) =  &\min\limits_{\pi\in \Pi(\alpha,\beta)} \int_{\cX \times \cY} c(x, y) d\pi(x, y)\label{eq: individual envy constraint primal}\\
    s.t. &\int_{\cX}c(x, y_j) d\pi(x, y_j) - \int_{\cX}c(x, y_j) d\pi(x, y_k) \frac{\beta_j}{\beta_k} \leq \lambda_j 
    &\forall (j,k)\in [n]^2, j\neq k\nonumber
\end{align}
Although envy constraints make the solution space more
complicated, we show that it retains the geometric structure of being the
intersection of half-spaces. The dual of \eqref{eq: individual envy constraint
primal} can be derived using Fenchel-Rockafellar's theorem:
\begin{align}
    \max_{g\in\bbR^n, \gamma\in\bbR^{n^2-n}_+}\cE(g, \gamma) \coloneqq &\sum_{j\in [n]}\int_{\bbL_{y_j}(g, \gamma)} \dualg(x, y_j) d\alpha(x)
     + g^\top \beta - \sum_{j,k, j\neq k}\gamma_{jk}\lambda_j\label{eq: individual envy constraint dual laguerre}
\end{align}
where 
\begin{equation}
\label{eq: dual g with envy}
    \dualg(x, y_j) \coloneqq \left(1+\sum_{k\neq j}
\gamma_{jk}\right)c(x, y_j) - \sum_{k\neq j}\gamma_{kj}c(x,
y_k)\frac{\beta_k}{\beta_j} - g_j,
\end{equation}%
\begin{equation}
\label{eq: generalized laguerre cell}
    \bbL_y(g, \gamma) \coloneqq \left\{ x\in\cX: y = \argmin\limits_{y'\in\cY} \dualg(x, y')\right\}.
\end{equation}%

\begin{theorem}
    If $\alpha$ is a continuous measure, and $\beta$ a discrete measure, then
    $L(\alpha, \beta, \lambda) = \max\limits_{g, \gamma}\cE(g, \gamma)$, and
    the optimal solution $\pi$ of \eqref{eq: individual envy constraint
    primal} is given by the partition $\left\{\bbL_{y_i}(g^*, \gamma^*), i\in
    [n]\right\}$: $d\pi(x, y_i) = d\alpha(x)$ if $x\in \bbL_{y_i}(g^*, \gamma^*)$,
    $0$ otherwise.
    \label{thm: duality of envy constrained OT}
\end{theorem}
Note that when $c(x, y_i) = -x_i$, $\dualg(x, y_j)$ is linear in $x$,
which means that the new Laguerre cells $\bbL_y(g, y)$ given in Equation~\eqref{eq: generalized
laguerre cell} are still intersections of half spaces (some examples are given later in Figure~\ref{fig: slanted dist allocation plot}). Furthermore, the
allocation policy can be interpreted as a greedy policy based on the
adjusted utility given by \eqref{eq: dual g with envy}, which contains additional interaction terms that take envy into account.

\section{Stochastic Optimization}
\label{sec: stochastic optimization}
In Section~\ref{sec: solution structure} we showed that the optimal allocation policy to our problem has a simple geometric structure in the form of Laguerre cells. In this section we present a practical algorithm for actually computing the optimal Laguerre cells. First we show that
the objective function $\cE(g, \gamma)$ in \eqref{prob: optimal matching policy problem}
is concave by rewriting the objective as follows:
\begin{equation}
    \label{eq: dual objective integration version}
    \cE(g, \gamma) = \int_{\cX} \min_{i\in [n]} \dualg(x, y_i) d\alpha(x) + g^\top \beta - \sum_{j,k, j\neq k}\gamma_{jk}\lambda_j
\end{equation}
Since $\dualg(x, y_j)$ is linear in $g$ and $\gamma$ and taking a minimum preserves concavity, the objective function is
concave. 
Therefore, the dual problem is a constrained convex optimization problem. 
The gradient of $\cE(g, \gamma)$ can be computed as follows:
\begin{align}
    \nabla_g \cE(g, \gamma)_j &= -\int_{\bbL_{y_j}(g, \gamma)} d\alpha(x) + \beta_j\label{eq: dual g gradient}\\
    \nabla_\gamma \cE(g, \gamma)_{jk} &= \int_{\bbL_{y_j}(g, \gamma)} c(x, y_j)d\alpha(x)
     - \int_{\bbL_{y_k}(g, \gamma)} c(x, y_j)\frac{\beta_j}{\beta_k}\,d\alpha(x) - \lambda_j\label{eq: dual gamma gradient}
\end{align}

\begin{algorithm}
	\DontPrintSemicolon 
	\KwIn{Distribution $\alpha$, target matching distribution $\optp$, timesteps $T$.}
	Initialize $g_0 = 0, \gamma_0 = 0, \eta = \frac{1}{\sqrt{T}}$.\;
	\For{$t \gets 0, 1, 2, \ldots, T$} {
	    Sample $x_t \sim \alpha$\;
	    $g_{t+1} \leftarrow g_t + \eta  \hat\nabla_g \cE(g, \gamma)$\;
	    $\gamma_{t+1} \leftarrow \left(\gamma_{t} + \eta \hat\nabla_\gamma \cE(g, \gamma)\right)^+$\;
	}
	\KwRet{$\sum_{t=1}^T g_t/T, \sum_{t=1}^T \gamma_t/T$}
	\caption{Projected SGD for Envy Constrained Optimal Transport}
	\label{algo: Projected SGD for Envy Constrained Optimal Transport}
    \end{algorithm}
Calculating this gradient is hard, as it involves integration over an arbitrary
measure $\alpha$. However, an unbiased, stochastic version of the gradient can
be easily obtained from a single sample $x\sim \alpha$:
\begin{align}
    \hat\nabla_g \cE(g, \gamma)_j &= -\1[x \in \bbL_{y_j}(g, \gamma)] + \beta_j\label{eq: dual g gradient stochastic}\\
    \hat\nabla_\gamma \cE(g, \gamma)_{jk} &=  c(x, y_j)\1[x \in \bbL_{y_j}(g, \gamma)]
    - c(x, y_j)\frac{\beta_j}{\beta_k}\1[x \in \bbL_{y_k}(g, \gamma)]- \lambda_j\label{eq: dual gamma gradient stochastic}
\end{align}
The details of the algorithm is given in Algorithm~\ref{algo: Projected SGD for
Envy Constrained Optimal Transport}.  Standard projected SGD analysis (see for
example \cite{harvey2018machine}) tells us that Algorithm~\ref{algo: Projected
SGD for Envy Constrained Optimal Transport} converges at the rate
$\cE(g^*, \gamma^*)-\cE\left(\bbE[g_T], \bbE[\gamma_T]\right)\leq
O\left(\frac{1}{\sqrt{T}}\right)$.

\section{Learning from Samples}
\label{sec: learning from samples}
So far we have assumed that the true underlying distribution is
known, and that we can freely draw independent samples it.
In many settings, we only have access to $\alpha$
in the form of finite number of i.i.d. samples. The goal of this section is to establish a sample complexity bound for solving the dual problem
$\eqref{eq: individual envy constraint dual laguerre}$.


In this section, we focus only on the assignment cost function $c(x, y_i) =
-x_i$, which models our original resource allocation problem proposed in Section~\ref{sec: problem formulation}.
Let $S = \{X^1, X^2, \ldots, X^m\}$ be $m$ independent samples from
$\alpha$. The empirical version of the dual objective \eqref{eq: dual objective
integration version} is:
\begin{equation}
    \label{eq: dual objective empirical}
    \cE_S(g, \gamma) = \frac{1}{m} \sum\limits_{t=1}^m\min_{i\in [n]} \dualg(X^t, y_i) + g^\top \beta - \sum_{j,k, j\neq k}\gamma_{jk}\lambda_j
\end{equation}
Let $\hat g_S, \hat \gamma_S$ be the empirical maximizer given the set of
samples $S$: $(\hat g_S, \hat\gamma_S) \coloneqq \argmax \cE_S(g, \gamma)$, and
$g^*, \gamma^*$ be the population maximizer $(g^*, \gamma^*)=\argmax \cE(g,
\gamma)$.
We want to bound the number of samples needed so that $\cE(g^*,
\gamma^*) - \cE(\hat g_S, \hat \gamma_S)$ is small with high probability.  Let's
introduce some notations to facilitate our later discussions.  Define the
following hypothesis class for each $i$:
\begin{align}    
    F_i = \left\{x\mapsto \dualg(x, y_i) + g^\top \beta - \sum_{j,k, j\neq k}\gamma_{jk}\lambda_j: g \in \bbR^n, \gamma \in \bbR_+^{n(n-1)}\right\}.
    \label{eq: separate hypothesis class definition}
\end{align} 
as well as the overall hypothesis class:
\begin{align}    
    F = \left\{x\mapsto \min_{i\in[n]}\dualg(x, y_i) + g^\top \beta - \sum_{j,k,
    j\neq k}\gamma_{jk}\lambda_j : g \in \bbR^n, \gamma \in
    \bbR_+^{n(n-1)}\right\}.
    \label{eq: hypothesis class definition}
\end{align} 
Plugging $c(x, y_i) = -x_i$ into the definition of $\dualg$, we see that for a given $g, $ and $\gamma$, the corresponding
hypothesis $f_i\in F_i$ can be written as $f_i(x) = w^\top x + b$, where 
\begin{equation}
    \label{eq: w definition}
w_j = 
  \begin{cases}
    -(1+\sum\limits_{k\neq i}\gamma_{ik}), &\text{if $j=i$}\\
    \gamma_{ji}\frac{\beta_j}{\beta_i}, &\text{if $j\neq i$}
  \end{cases}, \text{ and} 
\end{equation}
\begin{equation}
    \label{eq: b definition}
b = -g_i + g^\top \beta -  \sum_{j,k,
j\neq k}\gamma_{jk}\lambda_j.
\end{equation} 
It also follows that 
\begin{equation}
    \label{eq: hypothesis class contained in the min over halfspaces}
    F \subseteq F_{min} \coloneqq \{x\mapsto \min_i f_i(x): f_i\in F_i\}.
\end{equation}
Note that $F$ defined in \eqref{eq: hypothesis class definition} is the main object of interest, as it contains all the possible Laguerre cell parameters.  We showed in \eqref{eq: hypothesis class contained in the min over halfspaces} that $F$ is at most as complex as $F_{min}$, a hypothsis class constructed from $n$ affine hypothesis classes. This interpretation of the original hypothesis class as the minimum over $n$
affine hypothesis classes is the key observation to prove the sample complexity
bound.  We prove our main result under the following boundedness assumption:
\begin{assumption}
    The hypothesis $f(x) = \min_i f_i(x) = \min\limits_i {w^i}^\top x +b^i$ corresponding to the optimal dual solution $g^*,
    \gamma^*$ satisfies $\|w^i\|_1\vee |b^i| \leq R$ for some $R>0$. In particular,
    these assumptions imply that $F_i$ and $F$ are uniformly bounded by $R\bar x + R$.
    \label{assum: boundedness assumptions}
\end{assumption}
From \eqref{eq: w definition} and \eqref{eq: b definition} we can see that this
is essentially a bound on the optimal dual variables $g^*, \gamma^*$, and a
bound on the ratio $\beta_j/\beta_i$, both of which are determined by the input
distributions $\alpha, \beta$, and do not depend on the number of samples. In
other words, $R$ is a problem dependent constant.

\begin{theorem}
Under Assumption~\ref{assum: boundedness assumptions}, for a given sample size $m$, 
with probability $1-\delta$,  $\cE(g^*, \gamma^*)
- \cE(\hat g_S, \hat \gamma_S) < O\left(\sqrt{\frac{(\log m)^3 + \log(1/\delta)}{m}}\right)$. 
\label{thm: sample complexity bound}
\end{theorem}
The proof of this theorem can be found in the appendix. We will provide some high level intuition here. The proof uses the fat-shattering dimension, which is a concept that generalizes the Vapnik–Chervonenkis (VC) dimension to functions of real values. Like the VC dimension, the fat-shattering dimension is a measure of capacity for a class of functions. The larger the fat-shattering dimension, the more ``complex'' the function class is. Intuitively, the more complex the function class is, the more samples one needs to accurately identify the optimal function within that class. Recall that the object of interest in our paper is the Laguerre cell, whose boundaries are defined by a set of hyperplanes. Readers readers familiar with traditional PAC learning results might recall that hyperplanes have a low VC dimension. It turns out that the hypothesis class associated with Laguerre cells defined above in \eqref{eq: hypothesis class definition} has a low complexity as well. 
Similar to how low VC dimension leads to low sample complexity in classification tasks, the fact the boundaries of our Laguerre cells are consisted of hyperplanes also lead to low sample complexity in our setting.

\paragraph{Relation to Section~\ref{sec: stochastic optimization}} Note that there is some connection between the convergence of a stochastic optimization method (computational complexity), and the sample complexity bound of the same problem (learning complexity). Both say something about the number of steps/samples one needs to arrive at a good solution. However, in general, these two things are not the same. In particular, the learning complexity result from this section is algorithm agnostic. The proof of Theorem~\ref{thm: sample complexity bound} actually implies uniform convergence: that no matter which hypothesis one considers, its' empirical objective will be close to the expected objective. However, Theorem~\ref{thm: sample complexity bound} does not provide a way for practitioners to actually compute the solution. The result from Section~\ref{sec: stochastic optimization} on the other hand does provide an algorithm for practitioners to use, and shows that the algorithm is computationally efficient. However, the convergence result  in Section~\ref{sec: stochastic optimization} only works for the specific optimization method that we proposed. Therefore these two results are complements of each other, and together paint a relatively complete picture on how difficult the task it.

\section{Experiments}
\label{sec: experiments}
\begin{figure*}[]
	\begin{minipage}{0.33\textwidth}
	    \centering
	    \includegraphics[width=1\textwidth]{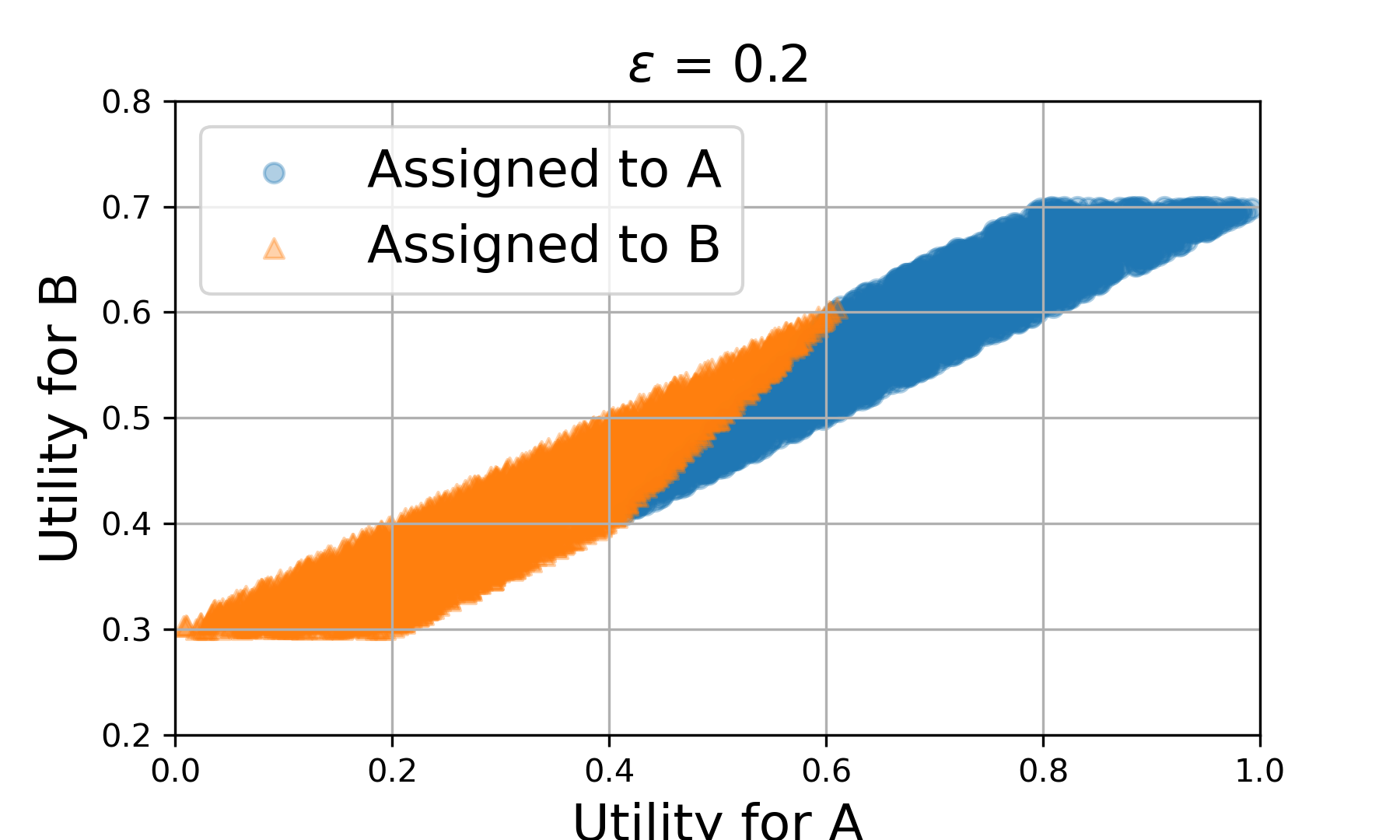}
	\end{minipage}%
	\begin{minipage}{0.33\textwidth}
	    \centering
	    \includegraphics[width=1\textwidth]{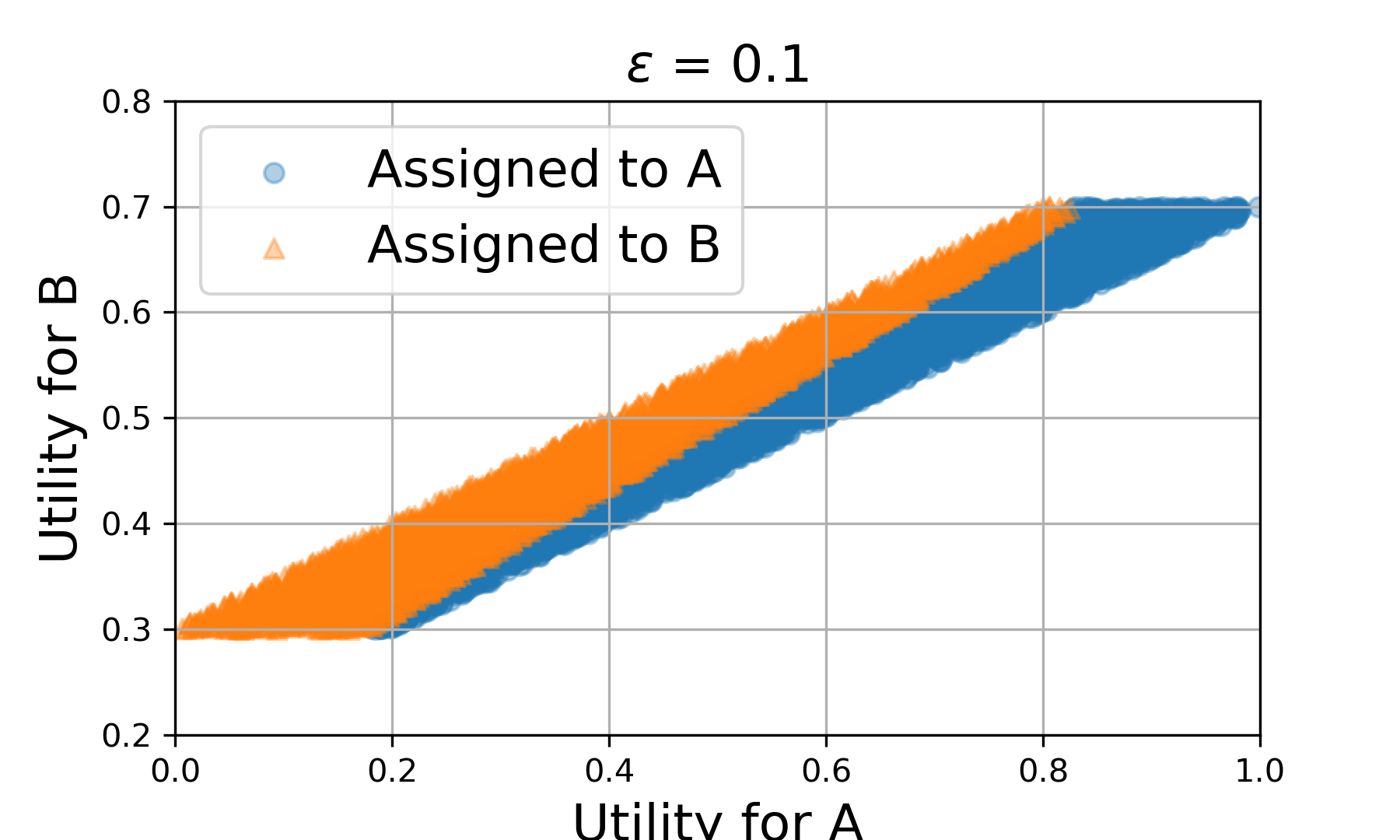}
	\end{minipage}%
	\begin{minipage}{0.33\textwidth}
	    \centering
	    \includegraphics[width=1\textwidth]{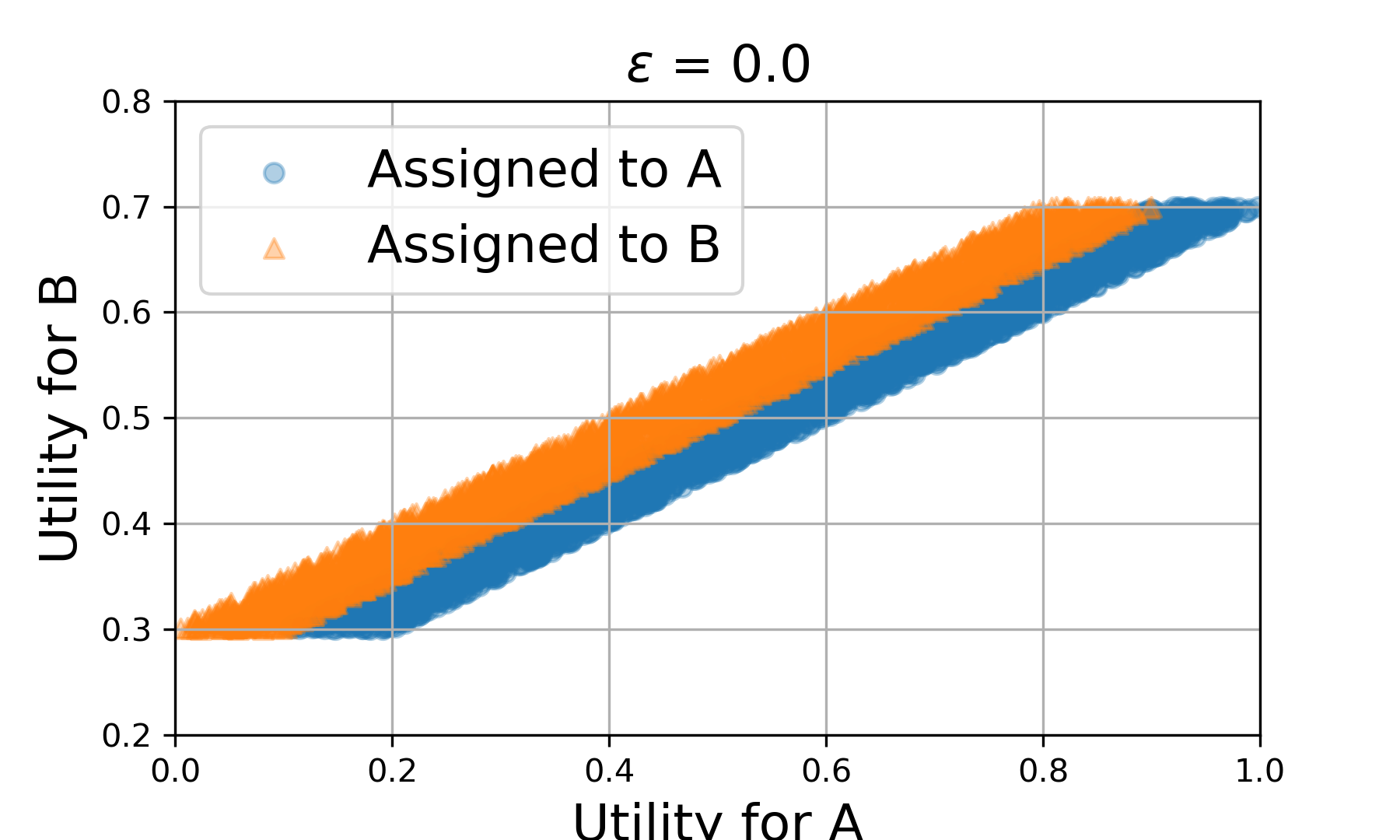}
	\end{minipage}
	\caption{Allocation policy for artificial data under different envy
	constraints. From left to right: $\epsilon=0.2, 0.1, 0.0$. When the envy
	constraint is loose (large $\epsilon$), $B$ envies $A$, since both agents
	prefer the items on the top right, but most of them are allocated to $A$. 
	As the envy constraint tightens, the allocation boundary tilts in the
	direction that makes the allocations more even between the two agents.
	}
	\label{fig: slanted dist allocation plot}
\end{figure*}
\begin{figure*}
	\begin{minipage}{0.48\textwidth}
    \centering
    \includegraphics[width=0.8\columnwidth]{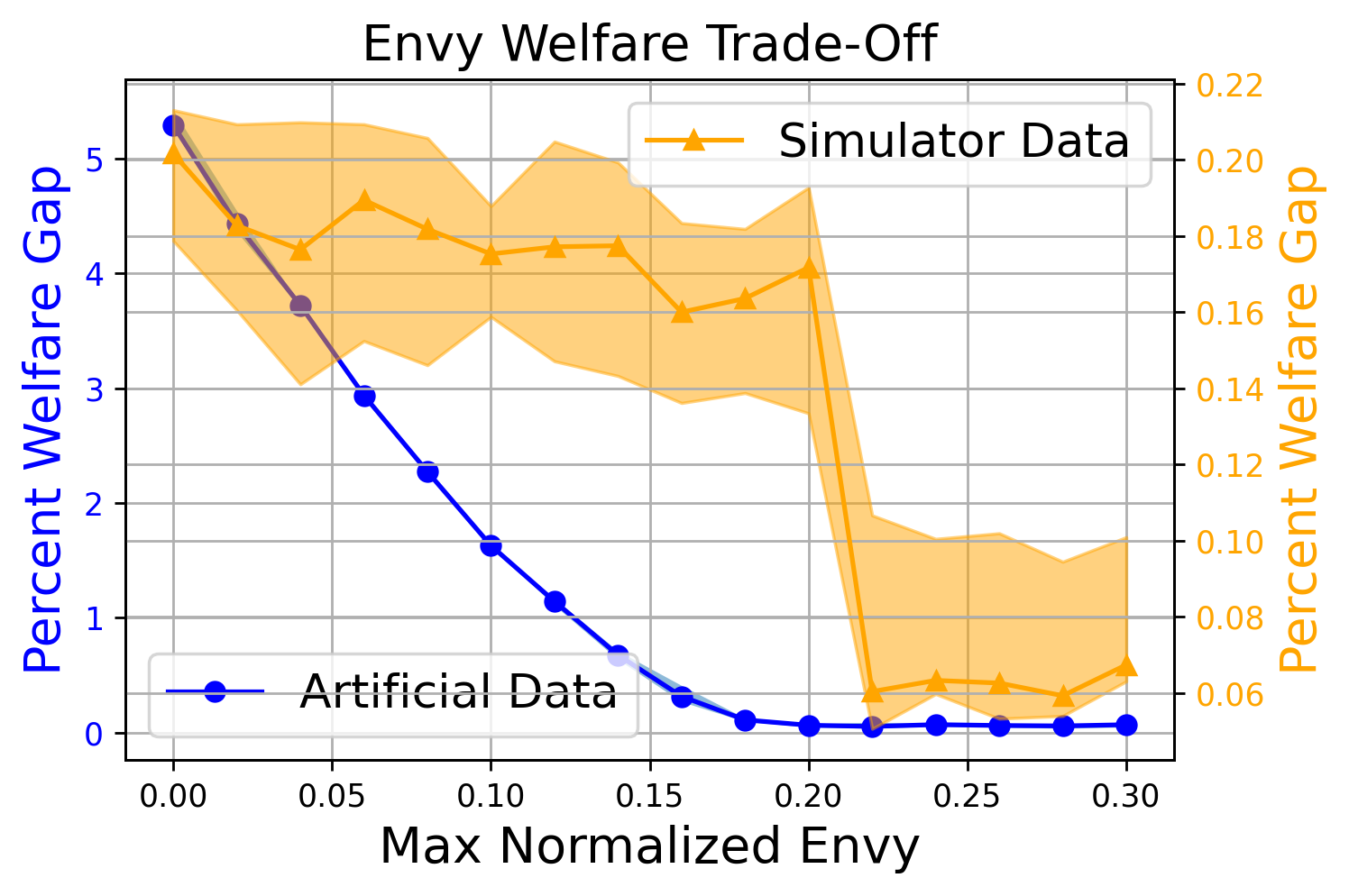}
    \caption{The trade-off curve between envy and welfare for both data-sets.
    The shaded region is between 25th and 75th percentile of the trials.  The
    non-monotonicity in the plot for the simulator data is due to the
    stochasticity in the SGD algorithm. } 
    \label{fig: slanted dist trade off curve}
	\end{minipage}%
	\hfill
	\begin{minipage}{0.48\textwidth}
    \centering
    \includegraphics[width=0.8\columnwidth]{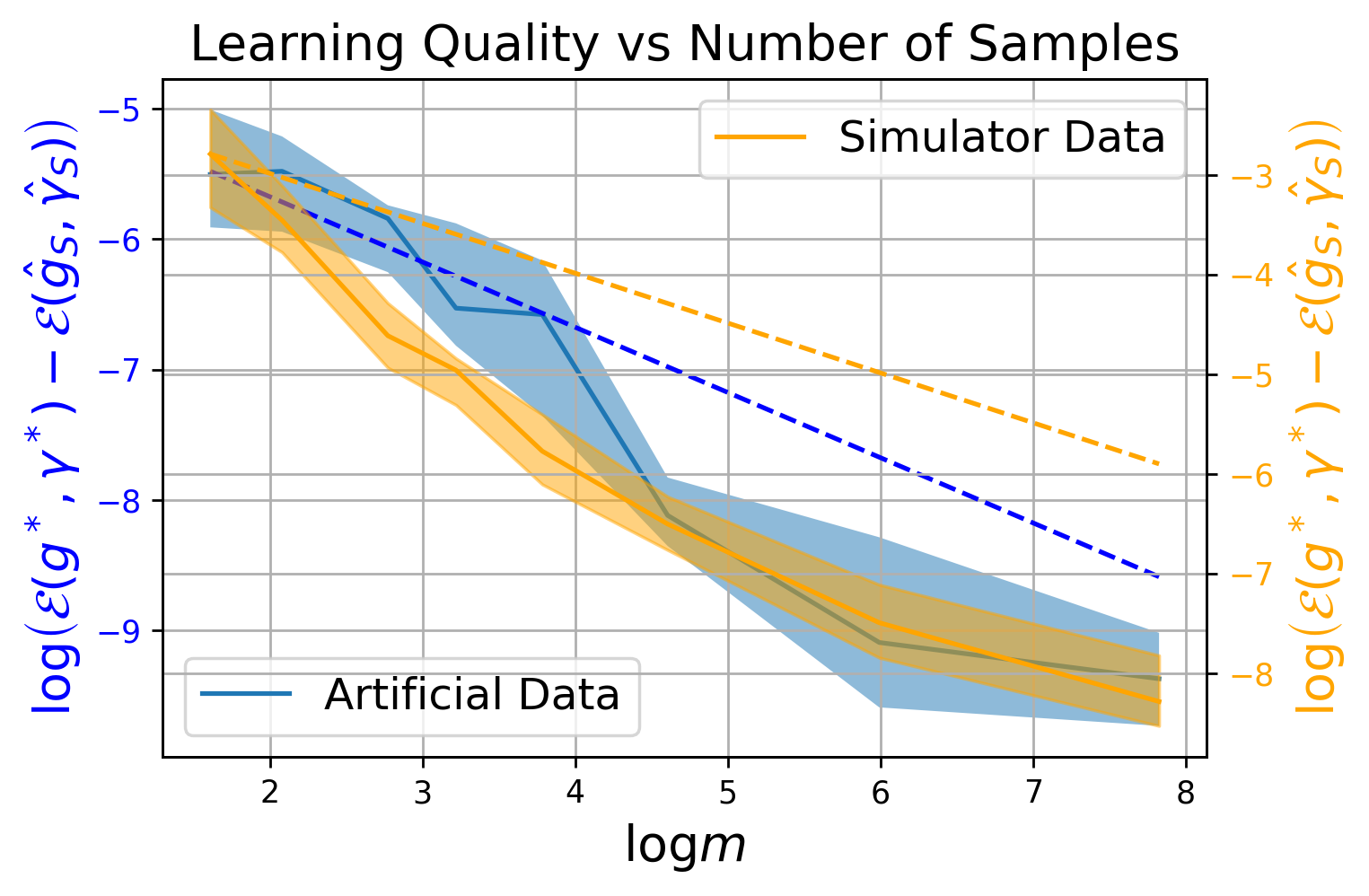}
    \caption{Approximation gap with respect to sample size. Both $x$ and $y$ axis are in log scale. The solid line is the median and the shaded region is between the $25$th and $75$th percentile.  The dashed lines show what the theoretical $1/\sqrt{m}$ rate would look like.}
    \label{fig: sample complexity plot}
	\end{minipage}
\end{figure*}
We test our solution with both artificial data, and simulated data from a realistic
simulator for blood donor matching developed by \cite{mcelfresh2020matching}. The artificial data contains two recipients, and their valuation distribution is a linearly transformed uniform distribution. This is to make visualization of the resulting allocation policy easier. The simulator data is based on geographical and population information from San Francisco, and contains $5$ recipients.
To set the envy budgets $\lambda\in\bbR^n$, we first decide on a constant
$\epsilon \in \bbR_+$, and then multiply this by the target matching
distribution $\optp \in \bbR^n$: $\lambda_{ij} = \epsilon \optp_i\,\forall
i,j$. 
With this setup $\epsilon$ is a bound on the 
normalized envy for each recipient: $\frac{1}{\optp_i}Envy(i)\leq \epsilon \,\forall i.$
Figure~\ref{fig: slanted dist allocation plot} illustrates how the allocation policy changes as we change $\epsilon$. As the envy constraint tightens, the decision boundary tilts in the direction that split the ``good'' (items which both agents prefer) and ``bad'' (items which both agents dislike) items more evenly between the recipients.

Next we investigate the trade-off between envy and social welfare by using SGD to
compute approximately optimal allocations for varying $\epsilon$.  We plot the
percent welfare gap (difference between the maximum welfare without envy
constraints, and the welfare with envy constraints, divided by the former) with
respect to realized, max normalized envy.  Figure~\ref{fig: slanted dist trade
off curve} shows the result.  For the simulator data, the welfare gap is small
even with a no-envy constraint, which means that aiming for envy free
allocations might make sense. In the case of the artificial data however, paying
$~50\%$ of the full price of fairness reduces $~65\%$ of the envy. In such
settings, one might want to sacrifice some envy for better welfare.

These experiments also highlight when envy arises. When recipients' utilities are
highly correlated, but one recipient has larger variance than others, that
recipient receives almost all the good items (which results in large envy for
other recipients), even though others value the items almost as much. In such
cases, a small reduction in welfare can reduce a large amount of envy. This
seems to be the case for the simulator data. On the other hand, if utilities are
correlated, but only one recipient has very strong preferences, then allowing a
small amount of envy can improve the welfare significantly.

Finally, in Figure~\ref{fig: sample complexity plot} we investigate the quality
of the empirical solutions as the sample size increases. It can been seen that
the approximation gap decreases faster than the theoretical rate, confirming our
sample complexity bound in Theorem~\ref{thm: sample complexity bound}. 

\section{Limitations and Future Directions}
\label{sec: future directions}
Although we believe that the model proposed here is natural, and captures the
most salient aspects of some of the resource allocation problems in real life, 
any implementation of our proposed strategy in critical applications such as 
blood donation should be prefaced with more rigorous backtesting in order to
minimize the risk of unintended consequences in application specific metrics not
studied in this paper.

For future directions, one key property that we did not study in this paper is
the problem of incentive compatibility. For our motivating application of blood
donation, this is not an issue because online platforms such as Meta has
proprietary models that can predict the matching quality between donor and
recipient. This means that the platform observes the value of matchings without
having to rely on the recipients to self-report. This is also true in many other
online matching problems such as sponsored ads. However, in settings where the
central planner relies on the recipients to self-report their valuations for each
of the items, incentive compatibility becomes a crucial issue. We are excited
about the potential of using Optimal Transport in fair-division, and plan on
exploring the incentive issues in future work.

\printbibliography

\newpage
\appendix
\section{Auxiliary Proofs}
\label{sec: auxiliary proofs}

\subsection{Proof of Theorem~\ref{thm: duality of envy constrained OT}}
\begin{proof}
	
Using Fenchel-Rockafellar's duality theorem, the dual of \eqref{eq: individual
envy constraint primal} can be written as 
\begin{align}
    &\max_{f, g, \gamma\geq 0} \int_\cX f(x) d\alpha(x) + g^\top \beta - \sum_{j,k, j\neq k}\gamma_{jk}\lambda_j\label{eq: individual envy constraint dual}\\
    s.t. \quad& \left(1+\sum_{k\neq j} \gamma_{jk}\right)c(x, j) - \sum_{k\neq j}\gamma_{ky}c(x, k)\frac{\beta_k}{\beta_j}\nonumber\\
    & - f(x) - g_j\geq 0 \quad\forall x\in \cX, y_j\in \cY\nonumber
\end{align}

Fixing $g\in \bbR^n$ and $\gamma \in \bbR^{n(n-1)}$, we can check using first
order conditions that the optimal $f(x)$ has the closed form expression:
$$\min\limits_{j\in[n]}\dualg(x, y_j) \coloneqq  \left(1+\sum_{k\neq j}
\gamma_{jk}\right)c(x, j) - \sum_{k\neq j}\gamma_{kj}c(x,
k)\frac{\beta_k}{\beta_j} - g_j$$
Using this, the infinite dimensional optimization problem in \eqref{eq:
individual envy constraint dual} can be transformed to a finite dimensional
optimization problem:
\begin{align}
    \max_{g, \gamma\geq 0}\cE(g,\gamma)\coloneqq\int_\cX  \min_{j\in [n]}\dualg(x, y_j) \,d\alpha(x) + g^\top \beta - \sum_{j,k, j\neq k}\gamma_{jk}\lambda_j\label{eq: individual envy constraint dual finite dim}
\end{align}
Alternatively, we can adapt the Laguerre cell notation in \eqref{eq: optimal
transport dual} to \eqref{eq: individual envy constraint dual finite dim}:
\begin{align}
    \cE(g, \gamma) = \sum_{i\in[n]}\int_{\bbL_{y_i}(g, \gamma)} \dualg(x, y_i)d\alpha(x) + g^\top \beta - \sum_{j,k, j\neq k}\gamma_{jk}\lambda_j\nonumber
\end{align}
where $\bbL_{y_i}(g, \gamma) = \left\{ x\in\cX: y_i = \argmin\limits_{y_j}
\dualg(x, y_j)\right\}$.
\end{proof}


\subsection{Proof of Theorem~\ref{thm: sample complexity bound}}
\begin{proof}
    We prove the result via uniform convergence:
\begin{align}
    &\cE(g^*,\gamma^*) - \cE(\hat g_S, \hat\gamma_S)\nonumber\\
    =   &\cE(g^*,\gamma^*) - \cE_S(\hat g_S, \hat\gamma_S) + \cE_S(\hat g_S, \hat\gamma_S) - \cE(\hat g_S, \hat\gamma_S)\nonumber\\
    \leq&\cE(g^*,\gamma^*) - \cE_S(g^*,\gamma^*) + \cE_S(\hat g_S, \hat\gamma_S) - \cE(\hat g_S, \hat\gamma_S)\nonumber\\
    \leq&\sup_{g, \gamma}\left(\cE(g, \gamma) - \cE_S(g, \gamma)\right) + \sup_{g, \gamma} \left(\cE_S(g, \gamma) - \cE(g, \gamma)\right)\nonumber\\
    \leq&2\sup_{g, \gamma}|\cE(g, \gamma) - \cE_S(g, \gamma)|\label{eq: uniform bound decomposition}
\end{align}
Clearly, it suffices to show that $\cE_S(\cdot)$ converges uniformly to $\cE(\cdot)$.  
For a given $g, \gamma$, the dual objective
function and its' empirical version can be written as 
$$\cE(g,\gamma) = \bbE_{\alpha}[f(X)], \quad \cE_S(g, \gamma) = \frac{1}{m}\sum_{t=1}^m f(X^t).$$
Then we can rewrite the supremum in \eqref{eq: uniform bound decomposition} as:
\begin{equation}
    \sup_{g, \gamma}|\cE(g, \gamma) - \cE_S(g, \gamma)| = \sup_{f\in F} \left|\bbE_\alpha [f(X)] - \frac{1}{m} \sum_{X\in S} f(X)\right|
\end{equation}

Since $|f(X)|\leq \Fbound$ for all $f\in F, X\in \cX$, it follows from
Theorem~26.5 in \cite{shalev2014understanding} that with probability $1-\delta$,
\begin{align}
    \sup_{f\in F} \bbE_\alpha [f(X)] - \frac{1}{m} \sum_{X\in S} f(X) 
    \leq 2\bbE_{S} \left[\Rad_m(F\circ S)\right] + \Fbound\sqrt{\frac{2\log(2/\delta)}{m}}
    \label{eq: partial uniform bound 1}
\end{align}
and the same also holds by replacing $F$ with $-F$.
Here
\begin{equation*}
    \label{eq: Rad definition}
    \Rad_m(F\circ S) \coloneqq \bbE_{\sigma} \left[\frac{1}{m}\sup\limits_f
    \sum_{j=1}^m \sigma_j f(X_j)\right] 
\end{equation*}
is the standard definition of Rademacher complexity of the set $F\circ S$. Since
$\sigma_i$ are $i.i.d.$ Rademacher random variables, it is easy to see that
$\Rad_m(F\circ S) = \Rad_m(-F\circ S)$.
Therefore we
can use a union bound to obtain that with probability $1-\delta$,
\begin{align}
\sup_{f\in F} \left|\bbE_\alpha [f(X)] - \frac{1}{m} \sum_{X\in S} f(X)\right|
\leq2\bbE_{S} \left[\Rad_m(F\circ S)\right] + \Fbound\sqrt{\frac{2\log(4/\delta)}{m}}
\label{eq: uniform bound}
\end{align}

It remains to bound the Rademacher complexity of the $F\circ S$. To do so, we 
use tools from learning theory, and give the following bound on the fat-shattering dimension (\cite{bartlett1996fat}) of the hypothesis class $F$.
\begin{lemma}
    Under Assumption~\ref{assum: boundedness assumptions}, $F$ has
    $\zeta$-fat-shattering dimension of at most
    $\frac{c_0\Fbound^2}{\zeta^2}n\log(n)$, where $c_0$ is some universal constant.
    \label{lem: shattering dimension bound}
\end{lemma}

The proof of Lemma~\ref{lem: shattering dimension bound} can be found in the Appendix. The above bound on the fat-shattering dimension can be used to bound the
covering number (see Definition~27.1 of \cite{shalev2014understanding}) of
$F\circ S$.  Theorem~1 from \cite{mendelson2003entropy} states that
\begin{equation}
\cN(\delta, F, ||\cdot||_2)\leq \left(\frac{2B}{\delta}\right)^{c_1 \fat_{c_2\delta}(F)}
\label{eq: covering number bound}
\end{equation}
where $B$ is a uniform bound on the absolute value of any $f\in F$.
Let $B = \Fbound$, we have that 

\begin{align*}
&\Rad_m(F\circ S)\\
\leq& \inf_{\delta'>0} \left\{4\delta' + 12\int_{\delta'}^{B}\sqrt{\frac{\log \cN(\delta, F, ||\cdot||_2)}{m}} d\delta\right\} \\
\leq& \inf_{\delta'>0} \left\{4\delta' + 12\frac{\sqrt{c_1c_0}}{c_2}B\sqrt{\frac{n\log n}{m}}\int_{\delta'}^{B}\sqrt{\log\left(\frac{2B}{\delta}\right)} d\delta\right\} \\
 = & c'\sqrt{\frac{n\log n (\log m)^3}{m}}
\end{align*}
Where we used Dudley's chaining integral \cite{sridharan,dudley1967sizes}, 
Lemma~\ref{lem: shattering dimension bound} and \eqref{eq: covering number bound},
and setting  $\delta'=\frac{1}{\sqrt{m}}$ respectively.
Plugging the above back to \eqref{eq: uniform bound} and \eqref{eq: uniform
bound decomposition}, we see that with probability $1-\delta$, 
\begin{equation*}
    \cE(g^*, \gamma^*)
- \cE(\hat g_S, \hat \gamma_S) \leq c'\left(\sqrt{\frac{n\log n (\log m)^3}{m}} + \sqrt{\frac{1\log\frac{1}{\delta}}{m}}\right).
\end{equation*}%
Conversely, ignoring the log terms, $m$ needs to be at most on
the order of $\tilde O\left(\frac{n}{\epsilon^2}\right)$ in order for $\cE(g^*, \gamma^*)
- \cE(\hat g_S, \hat \gamma_S)$ to be bounded by $\epsilon$ with high probability.

\end{proof}
Proof of Lemma~\ref{lem: shattering dimension bound}
\begin{proof}
    Theorem~3 in \cite{kontorovich2021fat} shows that $\fat_\zeta(F_{min}) \leq
    \frac{c_0\Fbound^2}{\zeta^2}n\log{n}$.
    Since the shattering dimension is monotone in the size of the set, we are
    done.
    
\end{proof}
\subsection{Experimental Setup}
\label{sec: experimental setup}
For the artificial data, the value utility vectors are generated from $X = [1,
0.7] - Z \begin{bmatrix}0.2, &0\\ 0.8, &0.4\end{bmatrix}$ where $Z\sim Unif(0,
1) \times Unif(0,1)$.  For finding the optimal allocation policy on the
artificial data, we used Algorithm~\ref{algo: Projected SGD for Envy Constrained
Optimal Transport} with $T=2\cdot10^5$. For simulator data, we used $T =
2.5\cdot 10^6$. To generate Figure~\ref{fig: slanted dist allocation plot} we
sampled $6000$ points from the distribution and plotted them, colored by the
allocation. For Figure~\ref{fig: sample complexity plot}, for each $m$ we ran
$16$ trials, sampling a different set of $m$ data points as our training data
per trial. All experiments are run on a 2019, 6-core Macbook Pro laptop. The simulator code is open sourced by \cite{mcelfresh2020matching} at \url{https://github.com/duncanmcelfresh/blood-matching-simulations}, and also included in the supplementary material.

\end{document}